\documentclass[a4paper,conference]{IEEEtran}%
\usepackage{ifthen}
\usepackage[cmex10]{amsmath}
\usepackage{amsfonts}
\usepackage{amssymb}
\usepackage{graphicx}%
\providecommand{\U}[1]{\protect\rule{.1in}{.1in}}
\providecommand{\U}[1]{\protect\rule{.1in}{.1in}}
\providecommand{\U}[1]{\protect\rule{.1in}{.1in}}
\ifCLASSINFOpdf
\else
\fi
\newtheorem{theorem}{Theorem}[section]
\newtheorem{lemma}[theorem]{Lemma}
\newtheorem{proposition}[theorem]{Proposition}

\newenvironment{proof}[1][Proof]{\begin{trivlist}
\item[\hskip \labelsep {\bfseries #1}]}{\end{trivlist}}

\newcommand{\qed}{\nobreak \ifvmode \relax \else
\ifdim\lastskip<1.5em \hskip-\lastskip
\hskip1.5em plus0em minus0.5em \fi \nobreak
\vrule height0.75em width0.5em depth0.25em\fi}

\begin{document}

\title{Polymer Expansions for Cycle LDPC Codes}
\author{Nicolas Macris and Marc Vuffray\\LTHC, IC, EPFL,\\CH-1015 Lausanne, Switzerland\\nicolas.macris@epfl.ch, marc.vuffray@epfl.ch}
\maketitle

\begin{abstract}
We prove that the Bethe expression for the conditional input-output entropy of cycle LDPC codes 
on binary symmetric channels above the MAP threshold is exact in the large block length limit.
The analysis relies on methods from statistical physics. The finite size corrections to the Bethe expression
are expressed through a polymer expansion which is controlled thanks to expander and counting arguments.
\end{abstract}





\section{Introduction}

\label{introduction}

A few years ago Cherktov and Chernyak \cite{Cherktov-Chernyak} devised a
\textit{loop series} which represents the partition function
of a general vertex model as the product of
the Bethe mean field expression and a residual partition function over a system of
{\it loops}. In this representation all quantities are entirely expressible in
terms of Belief Propagation (BP) marginals or messages.
However it has not been clear so far if this representation leads to a
\textit{controlled series expansions for the log-partition}, in other words the free
energy. If this is the case it should hopefully allow to control the
difference between the true free energy and the Bethe free energy. 

The loop 
expansion has a potential interest in coding theory since Low-Density-Parity-Check (LDPC) and Low-Density-Generator-Matrix
(LDGM) codes on general binary-input memoryless symmetric (BMS) channels fit in the framework of (generalized)
vertex models. 
In this context free energy is just another name for conditional input-output Shannon entropy.
For these models it is believed that the Bethe formula for the conditional entropy/free energy is exact. However there is no general proof,
except for the cases of the binary erasure channel \cite{binary}, LDGM codes for high noise, and in special situations for LDPC codes
at low noise \cite{Kudekar-Macris-2}.

We consider cycle LDPC codes for {\it high noise} (above the MAP threshold) on 
the binary symmetric channel (BSC). We show that, {\it under the assumption 
that there exists a fixed point for the BP equations, the average conditional entropy/free energy 
is given by the Bethe expression}. 
The novelty of the approach is to turn the loop expansion
into a rigorous tool allowing to derive provably convergent {\it polymer expansions} \cite{Brydges}. 
Controlling the loop expansion is a non-trivial task because in most situations of interest
 the number of loops proliferates. 
For example, this is the case (for the system of fundamental cycles) in 
capacity approaching codes even under MAP decoding \cite{sason}.

\section{Loop and polymer representations}\label{loop}
Let
$\Gamma=(V,E)$ be a graph with vertices $a\in V$ of regular degree $d$ 
 and edges $ab\in E$. The symbol
$\partial a$ denotes the set of $d$ neighbors of $a$. 
In vertex models the degrees of freedom are spins $\sigma_{ab}%
\in\{-1,+1\}$ attached to each edge.
At each function node $a\in V$ we attach a non-negative
function $f_{a}(\sigma_{\partial a})$ depending only on neighboring variables
$\sigma_{\partial a}\equiv(\sigma_{ab})_{b\in\partial a}$. We study
probability distributions which can be factorized as
\begin{equation}
\mu_{\Gamma}\left(  \vec{\sigma}\right)  =\frac{1}{Z_{\Gamma}}\prod_{a\in
V}f_{a}\left(  \vec{\sigma}_{\partial a}\right)  ,~Z_{\Gamma}=\sum
_{\vec{\sigma}}\prod_{a\in V}f_{a}\left(  \vec{\sigma}_{\partial a}\right),
\label{proba}
\end{equation}
and their associated free energy $f_{n}\equiv\frac{1}{n}\ln Z_{\Gamma}$.

For each edge $ab\in E$ we introduce
two directed \textquotedblleft messages\textquotedblright\ $\eta_{a\rightarrow
b}$ and $\eta_{b\rightarrow a}$. For the moment these variables are
 arbitrary and are collectively denoted by $\vec{\eta}$.
One has the identity \cite{Cherktov-Chernyak}
\begin{equation}
f_{n}=\frac{1}{n}\ln Z_{\mathrm{Bethe}}(\vec{\eta})+\frac{1}{n}\ln
Z_{\mathrm{corr}}(\vec{\eta}).
\label{iden}
\end{equation}
The first term is the Bethe free energy functional,
\begin{align}
\ln Z_{\mathrm{Bethe}}(\vec{\eta})=  &  \sum_{a\in V}\ln\left(  \sum
_{\vec{\sigma}_{a}}f_{a}\left(  {\sigma}_{\partial a}\right)  \prod
_{b\in\partial a}e^{\eta_{b\rightarrow a}\sigma_{ab}}\right) \nonumber\\
&  -\sum_{ab\in E}\ln\left(  2\cosh\left(  \eta_{a\rightarrow b}%
+\eta_{b\rightarrow a}\right)  \right).  
\label{bethefunc}%
\end{align}
The \textquotedblleft partition function\textquotedblright\ in the second term
can be expressed as a sum over \textit{all} subgraphs (not necessarily
connected) $g\subset\Gamma$
\begin{equation}
Z_{\mathrm{corr}}\left(  \vec{\eta}\right)  =\sum_{g\subset\Gamma} K(g) 
\label{sum} 
\end{equation}
and $K(g) = \prod_{a\in
g}K_{a}$ with
\begin{equation}\nonumber
K_{a}=\sum_{\vec{\sigma}_{a}}p_{a}\left(  {\sigma}_{\partial a}\right)
\prod_{b\in\partial a\cap g}\sigma_{ab}e^{-\sigma_{ab}(\eta_{a\rightarrow
b}+\eta_{b\rightarrow a})} 
\end{equation}%
\begin{equation}\nonumber
p_{a}\left(  {\sigma}_{\partial a}\right)  =\frac{f_{a}\left(  \vec{\sigma
}_{\partial a}\right)  \prod_{b\in\partial a}e^{\eta_{b\rightarrow a}%
\sigma_{ab}}}{\sum_{\vec{\sigma}_{a}}f_{a}\left(  \vec{\sigma}_{\partial
a}\right)  \prod_{b\in\partial a}e^{\eta_{b\rightarrow a}\sigma_{ab}}}.
\end{equation}
It is well known that the stationary points of \eqref{bethefunc}
satisfy the BP fixed point equations,
\begin{equation}
\eta_{a\rightarrow c}=\frac{\sum_{\vec{\sigma}_{a}}\sigma_{ac}f_{a}\left(
\vec{\sigma}_{\partial a}\right)  \prod_{b\in\partial a}^{b\neq c}%
e^{\eta_{b\rightarrow a}\sigma_{ab}}}{\sum_{\vec{\sigma}_{a}}f_{a}\left(
\vec{\sigma}_{\partial a}\right)  \prod_{b\in\partial a}^{b\neq c}%
e^{\eta_{b\rightarrow a}\sigma_{ab}}}. 
\label{eqn bp for vm}
\end{equation}
Remarkably, for any solution of \eqref{eqn bp for vm}, $K(g)=0$
if $g$ contains a degree one node. Thus if $\vec{\eta}$ is a fixed point of
the BP equations then $Z_{\mathrm{corr}}(\vec{\eta})$ is given by the sum in
\eqref{sum} over $g\subset\Gamma$ with no degree one nodes. Such graphs
 are called \textit{loops} (see figure \ref{fig loop to polym}).

One can recognize that
$Z_{\mathrm{corr}}$ can be interpreted as the partition function of a system
of polymers. A loop $g\subset\Gamma$ can be decomposed into its connected
parts in a unique way as illustrated on figure \ref{fig loop to polym}.
Connected loops are called polymers and are generically denoted by
the letter
$\gamma$. The important point is that {\it by definition} the polymers do not
intersect.
\begin{figure}[h]%
\centering
\includegraphics[
height=2.00in,
width=3.00in
]%
{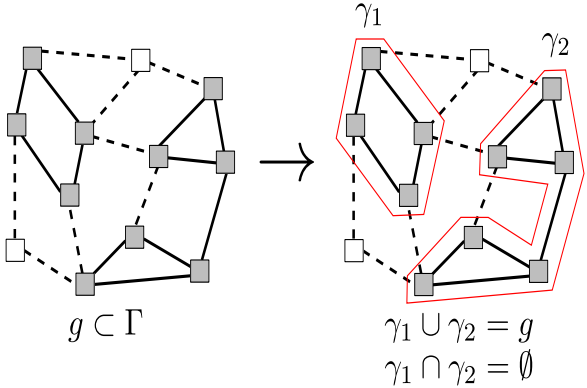}%
\caption{Left: an example of a loop graph $g$ with no dangling edge. Right:
decomposition of $g$ into its connected parts $\gamma_{i}$.}%
\label{fig loop to polym}%
\end{figure}
For each polymer $\gamma$ we define a weight (also called activity),
$K(\gamma)=\prod_{a\in\gamma}K_{a}$. Let $g=\cup_{i=1}^{M}\gamma_{i}$. Since
the $\gamma_{i}$ are disjoint, $\prod_{a\in g}K_{a}=\prod_{i=1}^{M}%
K(\gamma_{i})$. Thus equation \eqref{sum} can be cast in the form%
\begin{equation}
Z_{\rm corr}(\vec{\eta})=\sum_{M\geq0}\frac{1}{M!}\sum_{\gamma_{1},..,\gamma_{M}%
}\prod_{i=1}^{M}K\left(  \gamma_{i}\right)  \prod_{i<j}\mathbb{I}\left(
\gamma_{i}\cap\gamma_{j}=\emptyset\right).
\label{eqn zp}%
\end{equation}
In this sum each $\gamma_{i}$ runs over all connected subgraphs with no
dangling edges of the underlying graph $\Gamma$. The sum over the
number of polymers $M$ has a finite number of terms because the polymers
cannot intersect.

In the next paragraphs $\Gamma$ is a random $d$-regular graph. We 
denote by $\mathbb{P}$ and $\mathbb{E}$ the relevant probability and expectation over this ensemble.

\section{Polymer expansion}\label{mayer}

We wish to compute the correction to the Bethe free energy in
\eqref{iden}, namely $f_{corr}(\vec{\eta})\equiv\frac{1}{n}\ln Z_{\mathrm{corr}%
}(\vec{\eta})$ when $\vec{\eta}$ is a BP fixed point. Using \eqref{eqn zp} the
logarithm can be expanded as a power series in $K(\gamma_{i})$'s. This
yields the polymer (or Mayer) expansion \cite{Brydges}
\begin{align}
f_{corr}\left(  \vec{\eta}\right)   &  =\frac{1}{n}\sum_{M=1}^{\infty}\frac
{1}{M!}\sum_{\gamma_{1},..,\gamma_{M}}\prod_{i=1}^{M}K\left(  \gamma
_{i}\right) \nonumber\\
&  \times\sum_{G\subset\mathcal{G}_{M}}\prod_{\left(  i,j\right)  \in
G}\{-\mathbb{I}\left(  \gamma_{i}\cap\gamma_{j}\neq\emptyset\right)  \}.
\label{eqn mayer expansion poly}%
\end{align}
The third sum is over the set $\mathcal{G}_{M}$ of all 
{\it connected} graphs with $M$ 
vertices labeled by $\gamma_1,...,\gamma_M$, 
and at most one edge between each pair of vertices. 
The product of indicator functions is over edges
$(i,j)\in G$. It 
constrains the set of polymers $\gamma
_{1},...,\gamma_{M}$ to intersect according to the structure of $G$.
In this expansion one sums over an infinite number of terms so it is important
to address the question of convergence.
A criterion which ensures the convergence of the expansion {\it uniformly in system
size} $n$ (and thus ensures convergence in the infinite size limit) is
\begin{equation}
\sum_{t=0}^{+\infty}\frac{1}{t!}\sup_{a\in V}\sum_{\gamma\ni a}|\gamma
|^{t}|K(\gamma)|<1 
\label{criterion}
\end{equation}

To illustrate the use of the polymer expansion in a simple case,
consider a vertex model at high temperature defined by
\begin{equation}\nonumber
f_{a}(\sigma_{\partial a})=\frac{1}{2}(1+\tanh J_{a}\prod_{b\in\partial
a}\sigma_{ab})e^{\frac{1}{2}h_{ab}\sigma_{ab}}
\end{equation}
where $J_{a}$ and $h_{ab}$ are $\in\mathbb{R}$  with $\sup_{a\in V}J_{a}\equiv J<<1$
and $\sup_{ab\in E}h_{ab}<h<+\infty$. 
For $J$ is small enough the BP (\ref{eqn bp for vm})
equations have a unique fixed point solution  \cite{S. C. Tatikonda and M.
I. Jordan}. We call $\vec{\eta}_{n}^{\ast}$ this fixed point. The subscript
$n$ indicates (with some abuse of notation) that this fixed point depends on
the finite instance, that is,
the graph $\Gamma$, and $J_{a}$, $h_{ab}$. For the activities of the 
polymers, computed at the fixed point, we have the
bounds $\vert K(\gamma_{i})\vert\leq(2J)^{\vert\gamma_{i}\vert}$. Moreover the number of 
polymers $\gamma\ni a$ is (for each $a$) at most $e^{c_d \vert\gamma\vert}$ with $c_d\geq 0$ a numerical  constant
depending only on $d$. Using also that the smallest polymer must have $\vert\gamma\vert\geq 3$, it is then easily shown that the left hand side of \eqref{criterion} is $O(J^3) <<1$. By standard methods \cite{Brydges} one can then estimate the sum over $M$ in \eqref{eqn mayer expansion poly} term by term, which yields
\begin{equation}
|f_{\mathrm{corr}}\left(  \vec{\eta}_{n}^{\ast}\right)  |\leq (1 + O(J^3)) \frac
{1}{n}\sum_{a\in V}\sum_{\gamma\ni a}(2J)^{|\gamma|}e^{|\gamma|}
\label{ineg sur f}%
\end{equation}

\begin{proposition}\label{randomprop} 
For $J<J_0(h)$ small enough, we have
$\lim_{n\to+\infty}\mathbb{E}[f_{\mathrm{corr}}(\vec\eta_{n}^{*})] = 0$.
\end{proposition}
\begin{proof}[\it Proof idea]
From \eqref{ineg sur f} $\mathbb{E}[\vert f_{\mathrm{corr}}\left(  \vec{\eta}_{n}^{\ast}\right)\vert] \leq 
(1+O(J^3))\mathbb{E}[\sum_{\gamma\ni o}
(2J)^{\vert\gamma\vert} e^{\vert\gamma\vert}]$. Here $o$ is any specified node in the graph. In order to conclude it suffices to use the fact that on a random $d$-regular graph, with probability $1-o_n(1)$, polymers have a size
$\vert\gamma\vert\geq a_d\ln n$ ($a_d\geq 0$ a positive numerical constant).
\end{proof}

\section{Cycle LDPC codes over the BSC}

Random $d$-regular graphs
are equivalent to the ${\rm LDPC}(2,d)$ ensemble of cycle codes.
Code bits
$x_{ab}=0,1$ are attached to the edges $ab\in E$. In the spin language bits are
$\sigma_{ab}=\pm1$ and the parity check constraints are $\prod_{b\in\partial
a}\sigma_{ab}=1$. For definiteness we assume transmission over the ${\rm BSC}(p)$, $p\in\lbrack0,\frac{1}{2}]$.
Without loss of generality
one can assume that the transmitted word is $(1,...,1)$ so that MAP decoding
is based on the posterior distribution (\ref{proba}) with%
\begin{equation}
f_{a}(\sigma_{\partial a})=\frac{1}{2}(1+\prod_{b\in\partial a}\sigma
_{ab})\prod_{b\in\partial a}e^{\frac{1}{2}h_{ab}\sigma_{ab}}. 
\label{partit2}
\end{equation}
Here $h_{ab}$ is the half-log-likelihood for the bit $\sigma_{ab}=\pm1$, based 
on the channel output. The Shannon conditional input-output entropy and free energy are essentially equivalent, and related by the simple formula,
\begin{equation}
\frac{1}{n}H(\vec{X}|\vec{Y})=\mathbb{E}_{\vec{h}}[f_{n}(\vec{h})]-\frac
{1-2p}{2}\ln\frac{1-p}{p}
\label{shannon-free}
\end{equation}
where $\mathbb{E}_{\vec{h}}$ is the average over channel outputs (or the log-likelihood vector). 

We interested in the high noise regime where $p$ is close to $1/2$. Therefore
we seek solutions of the BP equations such that $\sup_{ab\in E}|h_{ab}|\leq h$
where $h>0$ is a fixed small number. We assume that for $h$ small enough there exists a fixed point of the BP equations
for each finite instance\footnote{This assumption can be relaxed by softening the hard constraint
in \eqref{partit2} and using existence results \cite{Yedidia}. Indeed all our estimates are uniform in the softening parameter. We omit this discussion here due to lack of space.}. We denote it $\vec{\eta}_{n}^{\ast}$ as before. Note that assuming its unicity is not needed.

\begin{proposition}\label{prop2}
Assuming the existence of a fixed point  $\vec{\eta}_{n}^{\ast}$ of the BP equations
for $h$ small enough, we have
$
\lim_{n\rightarrow+\infty}\mathbb{E}[\frac{1}{n}\ln Z_{\mathrm{corr}
}(\vec{\eta}_{n}^{\ast})]=0.
$
\end{proposition}

In view of \eqref{iden}, \eqref{shannon-free} the proposition implies that {\it the average conditional entropy is given by the average of the Bethe 
expression computed at the fixed point} (for $\vert p-\frac{1}{2}\vert << 1$).

In order to prove the proposition we will use the identity
\begin{equation}
\ln Z_{\mathrm{corr}}(\vec{\eta})=\ln Z_{p}(\vec{\eta
})+\ln \biggl\{  1+\sum_{\gamma\subset\Gamma}^{|\gamma|\geq
n/2}K(\gamma)\frac{Z_{p}(\vec{\eta}\mid\gamma)}{Z_{p}(\vec{\eta})}
\biggr\}
\label{split}
\end{equation}
where
\begin{align}
Z_{p}(\vec{\eta}\mid\gamma)=\sum_{M\geq0}\frac{1}{M!}\sum_{\gamma
_{1},...,\gamma_{M}}^{\mathrm{all|\gamma_{i}|<n/2}}  &  \prod_{i=1}%
^{M}K(\gamma_{i})\mathbb{I}(\gamma_{i}\cap\gamma=\emptyset)\nonumber\\
&  \times\prod_{i<j}\mathbb{I}(\gamma_{i}\cap\gamma_{j}=\emptyset)
\end{align}
and $Z_{p}(\vec{\eta})\equiv Z_{p}(\vec{\eta}\vert\emptyset)$. This identity is
derived by splitting the sum over $\gamma_{1},...,\gamma_{M}$ in \eqref{eqn zp}, into a
sum where all polymers are small ($\forall i,|\gamma_{i}|<n/2$), and a sum
where there exists at least one large polymer 
($\exists i,|\gamma_{i}|\geq n/2$); and by noting that when there exists a large polymer it has to be
unique. 

We will need three lemmas.
\begin{lemma}
\label{prop1} For $h$ small enough we have
$
\lim_{n\rightarrow+\infty}\mathbb{E}[\frac{1}{n}\ln Z_{p}(\vec{\eta
}_{n}^{\ast})]=0.
$
\end{lemma}
\begin{proof}[\it Sketch of Proof]
It is possible to estimate the activities computed at the fixed point,
\begin{equation}
\vert K(\gamma)\vert
\leq(1-\alpha_{d}\frac{d}{2}h^{2})^{n_{d}(\gamma)}\prod_{i=2}
^{d-1}(\alpha_{i}h^{d-i})^{n_{i}(\gamma)}.
 \label{activity-cycle}
\end{equation}
Here $0<\alpha_{d}<1$, and $\alpha_{i}>1$, $i=2,...,d-1$ are fixed numerical
constants (that we can take close to $1$). The $n_{i}(\gamma)$ denotes the
number of nodes of degree $i$ {\it in the polymer} $\gamma$.
Estimate \eqref{activity-cycle} is essentially optimal for $h$ small, as can be
checked by Taylor expanding $K(\gamma)$ in powers of $h_{ab}$. Hard
constraints manifest themselves in the factor $(1-\alpha
_{d}\frac{d}{2}h^{2})^{n_{d}(\gamma)}$ which {\it is not small enough} to compensate the 
entropic term $e^{c_{d}\vert\gamma\vert}$ in the convergence criterion. 
However for polymers of size $\vert\gamma
\vert<\frac{n}{2}$ we can use expander arguments to circumvent this problem. 
Let $e(g)$ the set
of edges in $E$ connecting $g$ to $\Gamma\setminus g$. We say that $\Gamma$ is
an expander if for all $g\subset\Gamma$ with $|g|\leq\frac{n}{2}$ we have
$|e(g)|\geq\kappa|g|$. For all $d\geq 3$ \cite{Bollobas}
\begin{equation}
\mathbb{P}[\Gamma\mathrm{{~is~an~expander~with~}\kappa=0.18\, d]=1-o_n(1)}
\label{expansion}
\end{equation}
Now note that for
polymers
$
e(\gamma)\leq\sum_{i=2}^{d-1}(d-i)n_{i}(\gamma)\leq d\sum_{i=2}^{d-1}
n_{i}(\gamma)
$. 
Therefore we deduce thanks to \eqref{expansion} that with high probability
$
\sum_{i=2}^{d-1}n_{i}(\gamma)\geq 0.18|\gamma|
$
and 
$
K(\gamma) \leq (2h)^{0.18\vert\gamma\vert}$
for 
$
\vert\gamma\vert <\frac{n}{2}
$
This is sufficient to control the convergence criterion,
and achieve the proof of this lemma by methods similarly to the high temperature case. 
\end{proof}

\begin{lemma}\label{lemmaratio} 
Fix $\epsilon>0$. Then
\begin{equation}\nonumber
\mathbb{P}[\forall\gamma\subset\Gamma:e^{-2n\epsilon}\leq\frac{Z_{p}(\vec
{\eta}\mid\gamma)}{Z_{p}(\vec{\eta})}\leq e^{2n\epsilon}]\geq1-\frac{1}{\epsilon}
o_n(1).
\end{equation}
\end{lemma}
The proof uses rather trivial bounds on the partition functions. We omit the details.
\begin{lemma}
\label{lemmalarge} Fix $\delta>0$. There exists a numerical
constant $C > 0$ and such that for $h$ small enough
\begin{equation}
\mathbb{P}\left[  \sum_{g\subset\Gamma}^{ \vert g\vert>n/2}\vert K\left(
g\right)  \vert\geq\delta\right]  \leq\frac{C}{\delta}e^{-n\alpha_{d}\frac
{d}{2} h^{2}}.
\label{boundonprob}
\end{equation}
This inequality is a fortiori valid for $g$'s replaced by $\gamma$'s in the sum.
\end{lemma}

\begin{proof}[\it Sketch of Proof]
We denote by $\mathcal{K}_{n}$ the complete graph with $n$ vertices. By
Markov's inequality,
\begin{align}
&\mathbb{P} \left[  \sum_{g\subset\Gamma}^{\left\vert g\right\vert
>n/2}|K\left(  g\right)  |\geq\delta\right]  \leq\frac{1}{\delta}%
\sum_{g\subset\mathcal{K}_{n}}^{\left\vert g\right\vert >n/2}\mathbb{E}[\left\vert K\left(  g\right)  \right\vert \mathbb{I}\left(
g\subset\Gamma\right)  ]
\nonumber\\&  
\leq\frac{1}{\delta}\sum_{g\subset\mathcal{K}_{n}}^{\left\vert g\right\vert
>n/2}\biggl\{(1-\alpha_{d}\frac{d}{2}h^{2})^{n_{d}(g)}
\prod_{i=2}^{d-1}(\alpha_{i}h^{d-i})^{n_{i}(g)}\biggr\}\mathbb{P}%
[g\subset\Gamma]
\label{prob}
\end{align}
Consider graphs $g$ with $n_{i}(g)$, $i=2,...,d$ fixed. Mackay \cite{MacKay}
 provides a bound for the probability\footnote{Here $[m]_i= m(m-1)...(m-i+1)$.}
$\mathbb{P}[g\subset\Gamma]$ of finding a particular subgraph into a regular
graph $\Gamma$. Namely for $\frac{1}{2}\sum_{i=2}^{d}in_{i}(g)+2d^{2}\leq
\frac{nd}{2}$,
\begin{equation}
\mathbb{P}[g\subset\Gamma]\leq\frac{\prod_{i=2}^{d}\left[  d\right]
_{i}^{n_{i}(g)}}{2^{\frac{1}{2}\sum_{i=2}^{d}in_{i}(g)}\left[  \frac{nd}%
{2}-2d^{2}\right]  _{\frac{1}{2}\sum_{i=2}^{d}in_{i}(g)}}.
\label{comb2}
\end{equation}
The number of subgraphs of $\mathcal{K}_{n}$ with given $n_{i}(g)$ is
\begin{align}
&  \frac{n!}{(n-\sum_{i=2}^{d}n_{i}(g))!\prod_{i=2}^{d}n_{i}(g)!}%
\nonumber\\
&  \times\frac{(\sum_{i=2}^{d}in_{i}(g))!}{(\frac{1}{2}\sum_{i=2}^{d}%
in_{i}(g))!2^{\frac{1}{2}\sum_{i=2}^{d}in_{i}(g)}\prod_{i=2}^{d}%
(i!)^{n_{i}(g)}}.
\label{combi}
\end{align}
Replacing \eqref{comb2} in \eqref{prob}, using \eqref{combi}, setting $x_{i}=\frac{n_{i}}{n}$, and 
performing an asymptotic calculation for $n$ large, we show (here $\vec x\equiv (x_2,...,x_d)$ and 
$\Delta \equiv \{\vec x\vert \frac{1}{2}\leq\sum_{i=2}^{d}x_{i}\leq1\}$)
\begin{align}
& \mathbb{P}\left[  \sum_{g\subset\Gamma}^{\left\vert g\right\vert
>n/2}|K\left(  g\right)  |\geq\delta\right] \leq\frac{1}{\delta}\int_\Delta%
d^d\vec x\,g_{n}(\vec x)\exp\bigl(n\{f_{n}(\vec x)\nonumber\\
&  +x_{d}\ln(  1-\alpha_{d}\frac{d}{2}h^{2})  +\sum_{i=2}%
^{d-1}x_{i}\ln(\alpha_{i}h^{d-i})  \}\bigr)
\label{exponent}
\end{align}
The large $n$ behavior of the integral asymptotic is controlled by $f_n(\vec x)$, and 
$g_n(\vec x)$ gives sub-dominant contributions that 
do not concern us here. We have
\begin{align}
&f_n(\vec x)   
=
\frac{1}{2}(  \sum_{i=2}^{d}ix_{i})  \ln\frac{1}{2}(  \sum_{i=2}^{d}
ix_{i})  
-  \sum_{i=2}^{d}x_{i}\ln\frac{ x_{i}}{\binom{r}{i}}
\nonumber \\ &
( 1-  \sum_{i=2}^{d}x_{i} )  \ln(
1-  \sum_{i=2}^{d}x_{i} ) 
- (  \frac{r}{2}-\frac{2r^{2}}{n})  \ln(  \frac{r}{2}%
-\frac{2r^{2}}{n})
\nonumber \\ &
+(  \frac{r}{2}- \frac{1}{2} \sum_{i=2}^{d}ix_{i}  -\frac{2r^{2}}
{n})
-
\ln(  \frac{r}{2}-\frac{1}{2}  \sum_{i=2}^{d}ix_{i}
-\frac{2r^{2}}{n})
\label{ass2}
\end{align}
\begin{figure}[h]%
\centering
\includegraphics[
height=2.00in,
width=3.00in
]
{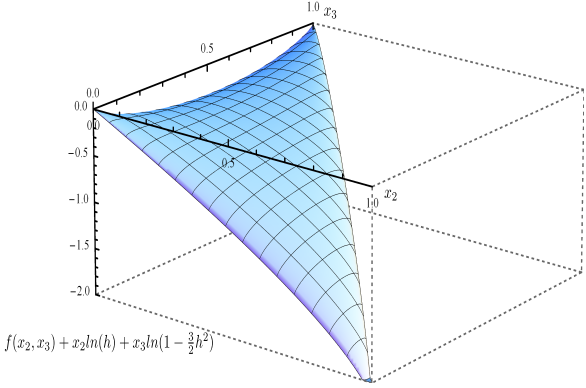}%
\caption{The exponent in \eqref{exponent} for $d=3$, for $h$ small enough, is strictly negative 
in the domain $\Delta$. Its maximum at $x_2=0, x_3=1$ is 
$O(h^2)$.}
\label{function}
\end{figure}
For $h$ small enough, in the domain $\Delta$, 
the exponent in \eqref{exponent} is strictly negative and attains its
maximum at the corner point $x_{2}=\cdots
=x_{d-1}=0$, $x_{d}=1$. At this point it is equal to 
$\ln(1-\alpha_d\frac{d}{2}h^2)$ which allows to conclude \eqref{boundonprob}.
\end{proof}

We are now in a position to prove proposition \ref{prop2}.

\begin{proof}[\it Proof of proposition \ref{prop2}]
In view of \eqref{split}, we must show that for
 $h$ small enough,
\begin{equation}
\frac{1}{n}\mathbb{E}\biggl\vert \ln\biggl\{ 1+\sum
_{\gamma\subset\Gamma}^{|\gamma|\geq n/2}K(\gamma)\frac{Z_{p}(\vec{\eta}
\mid\gamma)}{Z_{p}(\vec{\eta})}\biggr\}  \biggr\vert = o_n(1).
\label{preexp}
\end{equation}
Call $I_{\zeta}$ the event
\[
\sum_{\gamma\subset\Gamma}^{|\gamma|\geq n/2}|K(\gamma)|\frac{Z_{p}(\vec{\eta
}\mid\gamma)}{Z_{p}(\vec{\eta})}<\zeta
\]
where $\zeta$ is a positive constant that will be adjusted later on. We split
the expectation in two terms $A+B$ by conditioning over
$I_{\zeta}$ and its complement $I_{\zeta}^{c}$, and estimate each contribution.
For the first contribution, using $|\ln(1+x)|\leq|\ln(1-|x|)|$ for $|x|<1$,
$
A\leq\frac{1}{n}|\ln(1-\zeta)|\mathbb{P}[I_{\zeta}]\leq\frac{1}{n}%
|\ln(1-\zeta)|
$.
For the second contribution we have to estimate $\mathbb{P}[I_{\zeta}^{c}]$.
The events \{$\forall\gamma\subset\Gamma:e^{-2n\epsilon}\leq\frac{Z_{p}%
(\vec{\eta}\mid\gamma}{Z_{p}(\vec{\eta})}\leq e^{2n\epsilon}$\} and
\{$\sum_{\gamma\subset\Gamma}^{|\gamma|>n/2}|K(\gamma)|\leq\delta$\} imply
$I_{\delta e^{2n\epsilon}}$. Therefore $I_{\delta e^{2n\epsilon}}^{c}$ implies
the \textit{union of the complementary} events, so that applying the union
bound together with lemmas \ref{lemmaratio} and \ref{lemmalarge},
\begin{equation}\nonumber
\mathbb{P}[I_{\delta e^{2n\epsilon}}^{c}]\leq\frac{C}{\delta}e^{-n\alpha
_{d}\frac{d}{2}h^{2}}+\frac{1}{\epsilon}o_n(1).
\end{equation}
Now suppose for a moment that there exist a positive constant independent of
$n$ such that
\begin{equation}
\frac{1}{n}\biggl\vert\ln\biggl(1+\sum_{\gamma\subset\Gamma}^{|\gamma|\geq
n/2}K(\gamma)\frac{Z_{p}(\vec{\eta}\mid\gamma)}{Z_{p}(\vec{\eta}%
)}\biggr)\biggr\vert\leq C^{\prime\prime} 
\label{unif}%
\end{equation}
Then
$
B\leq C^{\prime\prime}\mathbb{P}[I_{\zeta}^{c}] 
$.
Setting $\zeta=\delta
e^{2n\epsilon}$, the above arguments imply
\begin{align}
&  \mathbb{E}\biggl[\frac{1}{n}\biggl\vert\ln\biggl(1+\sum
_{\gamma\subset\Gamma}^{|\gamma|\geq n/2}K(\gamma)\frac{Z_{p}(\vec{\eta}
\mid\gamma)}{Z_{p}(\vec{\eta})}\biggr)\biggr\vert\biggr]\nonumber\label{final}
\equiv A+B
\nonumber\\
&  \leq\frac{1}{n}|\ln(1-\delta e^{2n\epsilon})+\frac{C}{\delta}
e^{-n\alpha_{d}\frac{d}{2}h^{2}}+\frac{1}{\epsilon}o_n(1).
\nonumber
\end{align}
We are free to choose $\delta=e^{-n\alpha_{d}\frac
{d}{4}h^{2}}$ and $\epsilon=\alpha_{d}\frac{d}{16}h^{2}$ (lemmas
\ref{lemmaratio}
\ref{lemmalarge} hold) and this choice $A+B=o_n(1)$,
which proves \eqref{preexp}.

It remains to justify \eqref{unif}. From the convergence of the polymer
expansion we deduce that $\frac{1}{n}\ln Z_{p}(\vec{\eta}_{n}^{\ast})$ is
bounded uniformly in $n$. From \eqref{proba}, \eqref{partit2} we easily
show that $\frac{1}{n}\ln Z_{\Gamma}\leq\ln2+\frac{d}{2}h$. In the
high noise regime the BP messages are bounded so that from
\eqref{bethefunc} we deduce that $\frac{1}{n}\ln Z_{\mathrm{Bethe}}(\vec{\eta
}_{n}^{\ast})$ is bounded by a constant independent of $n$. 
Finally the
triangle inequality implies that
$
\frac{1}{n}|\ln Z_{\Gamma}-\ln Z_{\mathrm{Bethe}}(\vec{\eta}%
_{n}^{\ast})-\ln Z_{p}(\vec{\eta}_{n}^{\ast})|
$
is bounded uniformly in $n$. This is precisely the statement \eqref{unif}.
\end{proof}

\section{Conclusion}

The approach is quite general and can hopefully be generalized to standard 
irregular LDPC codes with bounded degrees and 
binary-input memoryless output-symmetric channels
with bounded log-likelihood variables. This will be the subject of future work.

\section*{Acknowledgment}

M.V acknowledges supported by the Swiss National Foundation for Scientific
Research, Grant no 2000-121903. N.M benefited from discussions with M.
Chertkov and R. Urbanke.



\begin{thebibliography}{9} 

\bibitem {Cherktov-Chernyak}M. Cherktov, V. Chernyak, \emph{Loop series for
discrete statistical models on graphs},
J. Stat. Mech., pp. 1-28, P06009, (2006).

\bibitem{binary}  C. Measson, A. Montanari, R. Urbanke, \emph{Asymptotic Rate versus Design Rate}, ISIT (2007) pp. 1541-1545.

\bibitem {Kudekar-Macris-2} S. Kudekar, N. Macris, \emph{Decay of Correlations
for Sparse Graph Error Correcting Codes}, SIAM J. Discrete Math. 25, pp.
956-988 (2011).

\bibitem {Brydges} D. Brydges, \emph{A short course on cluster expansions}, in
\emph{Ph\'{e}nomen\`{e}s critiques, syst\`{e}mes al\'{e}atoires et
th\'{e}ories de jauge}, K. Osterwalder and R. Stora ed. Les Houches, session XLIII, Part I, (1984).

\bibitem{sason} I. Sason, \emph{On universal properties of capacity-approaching LDPC code
ensembles}, IEEE Trans. on Information Theory, 55 pp. 2956-2990  (2009).

\bibitem {S. C. Tatikonda and M. I. Jordan} S. C. Tatikonda and M. I. Jordan,
\emph{Loopy belief propagation and Gibbs measures}, in Proc. 18th Annu. Conf.
Uncertainty in Artificial Intelligence (UAI-02), San Francisco, pp. 493-500 (2002).

\bibitem{Yedidia} J. S. Yedida, W. T. Freeman, and Y. Weiss, \emph{Understanding belief propagation and its generalizations}, in Exploring artificial intelligence in the new millennium, 
Morgan Kaufmann Publishers (2003). 

\bibitem{Bollobas} B. Bollobas, \emph{The isoperimetric number of
random regular graphs}, Euro. J. Combinatorics 9 241Ð244 (1988).

\bibitem {MacKay} B.D. McKay, Subgraphs of random graphs with specified
degrees, Congressus Numerantium 33, 213-223 (1981).
\end{thebibliography}
\end{document}